\documentclass{article}

\usepackage{latexsym,amsmath,amssymb,amsfonts,graphicx,amsthm,mathrsfs}
\usepackage[a4paper,margin=2.5cm,footskip=1cm]{geometry}

\usepackage{float}        
\usepackage{times}        
\usepackage[normalem]{ulem}
\usepackage{url}
\usepackage{braket}
\usepackage{mathtools}
\usepackage{authblk}
\usepackage{import}
\usepackage[usenames,dvipsnames]{color}

\newtheorem{thm}{Theorem}
\newtheorem{lem}[thm]{Lemma}

\newcommand{\q}[1]{| #1 \rangle}

\newcommand{\cl}[1]{e_{#1}}
\newcommand{\clt}[1]{e^\dag_{#1}}

\newcommand{\joint}{{\hat{p}}}
\newcommand{\Hi}{\mathcal{H}}

\renewcommand{\natural}{\mathbb{N}}

\newcommand{\pilim}{\bar{p}}

\newcommand{\Nd}{\mathcal{V}}

\newcommand{\Cd}{\mathcal{C}}
\newcommand{\Gr}{\mathcal{G}}
\newcommand{\X}{\mathcal{X}}

\newcommand{\U}{{\bf U}}
\renewcommand{\S}{{\bf C}}
\newcommand{\Swap}{{\bf S}}
\renewcommand{\P}{{\bf P}}
\newcommand{\Pu}{{\bf P^+}}
\newcommand{\Pd}{{\bf P^-}}
\newcommand{\integer}{\mathbb{Z}}

\renewcommand{\o}[1]{{\bf #1}}
\newcommand{\s}[1]{{\mathcal #1}}
\newcommand{\dos}{\mathcal{D}}
\newcommand{\dens}{\rho}
\newcommand{\pg}[2]{\mathbb{P}_{#2}(#1)}
\newcommand{\tr}{\mathrm{trace}}
\newcommand{\id}{\o{I}}
\newcommand{\qc}{\bf \Gamma}
\newcommand{\gc}{\o{\Psi}}
\newcommand{\p}{p}

\title{Simulation of Quantum Walks and Fast Mixing with Classical Processes}
\author[1]{Simon Apers \thanks{Corresponding author: simon.apers@ugent.be}}
\author[1,2]{Alain Sarlette \thanks{alain.sarlette@inria.fr}}
\author[3]{Francesco Ticozzi \thanks{ticozzi@dei.unipd.it}}
\affil[1]{Department of Electronics and Information Systems, Ghent University, Belgium}
\affil[2]{QUANTIC lab, INRIA Paris, France}
\affil[3]{Dipartimento di Ingegneria dell'Informazione, Universit\`a di Padova, Italy, and the Department of Physics and Astronomy, Dartmouth College, NH 03755, USA.}
\begin{document}

\maketitle

\begin{abstract}
We compare discrete-time quantum walks on graphs to their natural classical equivalents, which we argue are lifted Markov chains, that is, classical Markov chains with added memory. We show that these can simulate quantum walks, allowing us to answer an open question on how the graph topology ultimately bounds their mixing performance, and that of any stochastic local evolution. The results highlight that speedups in mixing and transport phenomena are not necessarily diagnostic of quantum effects, although superdiffusive spreading is more prominent with quantum walks.
\end{abstract}

Random walks are both ubiquitous models for natural processes and a powerful, versatile algorithmic tool to explore networks and extract information about their structure.
In recent years their quantum analogue, named quantum walks (QWs), was shown to hold similar promises. QWs describe the evolution of the position probability distribution of a ``walking'' quantum particle on a graph, possibly entangled with other quantum degrees of freedom (the so-called coin). The joint dynamics can be either discrete-time or continuous and must respect the graph locality \cite{aharonov1993quantum,farhi1998quantum,watrous1999quantum,kempe2003quantum}. Following the realization that QWs on a line can beat the diffusive behavior typical of classical stochastic processes \cite{ambainis2001one,aharonov1993quantum}, they have been invoked to explain improved transport phenomena in biological systems \cite{engel2007evidence,mohseni2008environment}, linked to thermodynamic theories, breakdown models and topological states of matter \cite{romanelli2014thermodynamics,oka2005breakdown,kitagawa2012observation}, and simulated in various experiments \cite{karski2009quantum,peruzzo2010quantum,genske2013electric,preiss2015strongly,flurin2017observing}.
Furthermore, they have been intensely studied as a paradigm for quantum computing \cite{childs2009universal,lovett2010universal} and to speed up algorithmic tasks \cite{ambainis2003quantum}, in particular, those related to the celebrated Grover search algorithm \cite{shenvi2003quantum,childs2004spatial,magniez2011search}.

Despite impressive advances in their analysis, elucidating the source and extent of quantum advantages from the perspective of QWs, as well as providing general design principles to ensure a quantum speedup, remain ongoing lines of research. A general quadratic speedup by QWs has been established for the hitting time \cite{childs2003exponential,szegedy2004quantum,kempe2005discrete,magniez2011search,krovi2016quantum,hoyer2016efficient}, thus searching for a marked node in a graph.
The complementary problem of mixing, that is, converging to a particular probability distribution over the nodes, has so far resisted a general QW speedup analysis, although it is closer to the original observation on the line \cite{ambainis2001one}. There is further evidence for a quadratic speedup with respect to classical Markov chains on specific graphs including the cycle \cite{aharonov2001quantum}, the hypercube \cite{moore2002quantum}, and the torus \cite{richter2007quantum}. A general characterization of QW mixing would be a fundamental step for investigating quantum \emph{vs.} classical differences in statistical mechanics (thermodynamic equilibration, transport phenomena, localization defects), and its algorithmic complexity is of key relevance for applications like sampling and Monte-Carlo simulations \cite{sinclair2012algorithms}. 

In this paper, we characterize mixing performance of QWs by showing that they belong to a class of processes which can be simulated by \emph{classical} Markov chains with additional finite memory, called ``lifted Markov chains'' (LMCs) \cite{chen1999lifting}. For general graphs, our constructive proof reminds a classical version of the ``Feynman clock Hamiltonians'' used to prove universality of adiabatic computing \cite{feynman1982simulating,kitaev2002classical,aharonov2008adiabatic}, in combination with ``stochastic bridges'' generalizing \cite{aaronson2005quantum} and \cite{pavon2010discrete,georgiou2015positive} to simulate quantum channels for fixed initial conditions. This allows us to derive a tight bound on potential QW mixing speedup, improving the known bounds from \cite{aharonov2001quantum,temme2010chi}. Furthermore, for lattices, on which most QW mixing speedups have been demonstrated, we relate the QWs to fast mixing LMCs that have not only the same mixing performance, but also the same structure \cite{diaconis2000analysis,diaconis2013spectral}, making them their natural classical analogue.

These results provide several insights. First, an observed speedup in mixing is {\em not} fundamentally diagnostic of a quantum effect, as it may always be explained by a purely classical memory. Second, QWs are essentially subject to the same bound on their mixing performance as other local processes, induced solely by the topology of the graph. Third, the search for a quantum advantage should focus on identifying efficient designs, in terms of the amount of memory or the graph knowledge required. For lattices, beating efficient classical algorithms is possible only for tasks beyond pure mixing. Whether for statistical mechanics, evolutionarily selected biological systems, or design of faster Monte Carlo algorithms, our results significantly narrow the context in which quantum effects may provide an intrinsic advantage.\\

\noindent \emph{QWs and their classical counterparts: a paradigmatic example. --} Usually, QWs are presented as the quantum analogues of, and compared to, classical random walks. We next argue that different classical models should be considered towards establishing an intrinsic quantum advantage in mixing, as QWs exhibit genuine {\em memory effects.} Standard discrete-time QWs \cite{meyer1996quantum,aharonov2001quantum} describe the evolution of the position distribution $p_t$ of a quantum particle (``walker'') over a discrete set of graph nodes $\Nd$. The quantum evolution of position is conditioned on additional degrees of freedom $\Cd$, the coin of the walker. The walker state is thus defined on the joint Hilbert space
$\Hi=\Hi_C \otimes \Hi_V = \text{span}\{{\q{c}\otimes\q{v}}|(c,v)\in \Cd\times \Nd\}.$ The cycle graph is a simple example where QWs provide a mixing speedup with respect to a classical walk, see Fig.~\ref{fig:cycle-walks}. 
To the nodes $\Nd = \{1,2,\dots,N\}$ of the cycle, the QW adds a binary coin $\Cd = \{+,-\}$, see \cite{ambainis2001one,aharonov2001quantum}.
Denoting $\P^\pm$ the cyclic permutation of position, that is, $\P^\pm \q{v} = \q{(v\pm1)\mathrm{mod }N}$ for $v\in {\cal V}$, the unitary QW primitive reads
\[
 \U = \Swap
        \, \left( \S \otimes {\bf I}_N \right)\, , \;\; \S = \begin{bmatrix} e^{-i\phi}\sqrt{1-\alpha} & e^{i\theta}\sqrt{\alpha} \\ -e^{-i\theta}\sqrt{\alpha} & e^{i\phi}\sqrt{1-\alpha} \end{bmatrix},
\]
where $\Swap=\ket{+}\bra{+}\otimes \Pu+\ket{-}\bra{-}\otimes \Pd$ expresses a conditional shift, while $\S$ is a general unitary ``coin toss'' on $\Hi_C$. The conditional motion can also be viewed as spin-orbit coupling. To actually mix, some decoherence or measurement rule must be added to this unitary evolution, see \cite{kendon2007decoherence} for a survey. For instance, after every application of $\U$, one can perform with probability $q$ a projective measurement in the canonical basis, after which the unitary evolution is resumed:
\begin{equation}\label{eq:Uex}
\ket{\psi_{t+1}}
     = \begin{cases} \U\ket{\psi_t} \text{ with probability } 1-q,\\
        \ket{c,v} \text{ with probability } q \, |\braket{c,v|\U|\psi_t}|^2 . \!\! \end{cases}
\end{equation}
A purely unitary QW is obtained with $q=0$, while $q=1$ projects the state on the reference basis at each step. The position distribution $p_t$ is obtained by tracing over the coin and considering the probabilities induced in the node basis at time $t$. The QW of Eq.~\eqref{eq:Uex} with, e.g., parameters $\alpha = 1/2, \phi=\theta=0$ and $q=O(1/N)$, converges towards a uniform $p_t$ in $t=O(N)$ steps, from any initial distribution \cite{ambainis2001one,aharonov2001quantum}. In contrast, a classical random walk over $\Nd$ with transition matrix $\P_0 = (\Pu + \Pd) / 2$ reaches the same distribution only after $O(N^2)$ steps.

\begin{figure}[!t]
\center
\includegraphics[width=80mm]{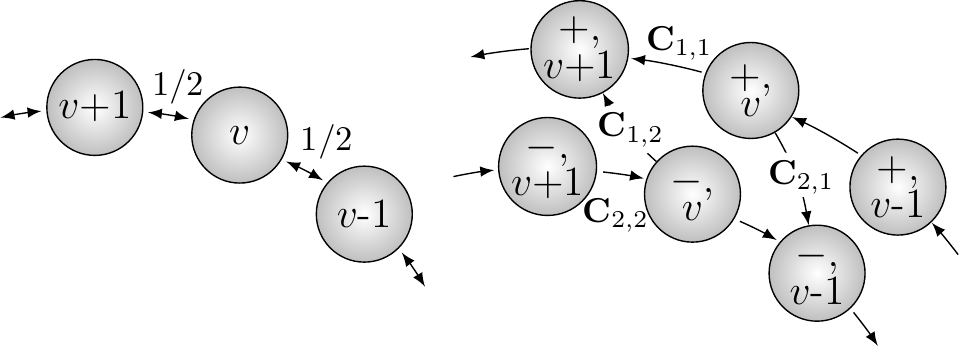}
 \caption{(left) random walk $\P_0$ on the $N$-cycle; (right) quantum walk unitary with coin toss $\S$, lifted Markov chain with a stochastic coin toss $\overline{\S}$, suggesting their comparison.}
 \label{fig:cycle-walks}
\end{figure}

Compared to a classical random walk, the QW above clearly adds memory via the coin degrees of freedom. Yet, QWs can exhibit memory effects even without coin. Consider the two-node graph without coin, $\mathcal{H}={\rm span}\{\ket{1},\ket{2}\},$ equivalent to a qubit, and take the Hadamard gate $U_H=(\sigma_x+\sigma_z)/\sqrt{2}$ as QW primitive. Starting on a given node, after one step, the distribution $p_1$ over $\ket{1},\ket{2}$ is uniform, yet at the second step the initial state is perfectly recovered since $U_H^2$ is the identity operator. This behavior, impossible for any classical Markov process on $\{1,2\}$, is due to the quantum state storing information in its relative phases, or coherences.
Hence, to establish if there is an intrinsic quantum advantage, QWs should be compared to classical local processes with at least a certain amount of additional memory.

Remarkably, a classical walker with memory that mixes fast on the cycle has already been proposed independently of the QW literature \cite{diaconis2000analysis,diaconis2013spectral}, and it shares striking similarities. This walker moves among {\em classical} states in $\Cd\times \Nd.$ 
Its probability distribution {$\joint$ over $\Cd\times \Nd$ evolves as $\joint_{t+1}=\P\, \joint_t,$} with stochastic transition matrix $\P$ having the same structure as $\U$, yet with $\S$ now replaced by a \textit{stochastic} coin toss:
\begin{equation}\label{eq:Pex}
  \P = \Swap
        \, \left( \bar{\S} \otimes {\bf I}_N \right)\, , 
        \quad \bar{\S} = \begin{bmatrix} 1-\alpha & \alpha \\ \alpha & 1-\alpha \end{bmatrix}.
\end{equation}
This can be seen as the mixture of two reversible evolutions: with probability $1-\alpha,$ the state follows the conditional shift $\Swap$; or, with probability $\alpha,$ the coin is switched before applying $\Swap.$ The coin allows the classical walker to retain and use information about its previous motion direction, in physical terms its momentum. The similarity between $\U$ and $\P$ carries a deeper connection, as $\P$ in Eq.~\eqref{eq:Pex} exactly describes the probabilistic evolution induced by Eq.~\eqref{eq:Uex} when starting with $\ket{\psi}=\ket{c,v},$ for some $(c,v) \in \Cd\times \Nd,$ and with $q=1$.

This $\P$ mixes over the cycle in $O(N)$ steps \cite{diaconis2000analysis,diaconis2013spectral}, provided $\alpha=O(1/N)$. This speedup, only due to classical memory, matches the one provided by the QW in Eq.\eqref{eq:Uex} with $q=O(1/N)$. In both cases, an $O(1/N)$ nonunitarity provides a good tradeoff between fast (deterministic) motion along the graph and losing correlation with the initial condition. From these observations, it appears most natural to compare QWs like Eq.\eqref{eq:Uex} to classical evolutions with memory like Eq.\eqref{eq:Pex}, which are formalized as LMCs \cite{chen1999lifting}.\\

\noindent \emph{QWs and LMCs as local processes with equivalent mixing performance. --} Consider a graph with node set $\Nd$ and edges ${\cal E} \subset \Nd \times \Nd$. The nodes could represent energy levels and the edges allowed transitions. The QW and LMC constructions both start by building a lifted graph, where each node of the initial graph is split into ``lifted nodes'' or ``sublevels''. This is done without loss of generality by introducing a coin set ${\Cd},$ defining the lifted nodes $\Cd\times \Nd = \{(c,v)\}$ and selecting lifted edges in $\{(\,(c,v),\,(c',v')\,)~|~(v,v') \in {\cal E} \}$, thus without introducing transitions that were not allowed before lifting.

A general QW is then described by a quantum channel over the space generated by viewing coin and node as quantum numbers, i.e.,
\begin{equation}\label{eq:QWgen}
\rho_{t+1} = \sum_k  {\bf M_k} \rho_t {\bf M_k}^\dagger,
\end{equation}
where $\rho_t$ is a density operator on $\Hi = \mathrm{span}\{\q{c,v}~|~(c,v)\in \Cd\times \Nd\}$ and the $\bf M_k$ satisfy $\bf \sum_k M_k^\dagger M_k = {\bf I}_{\Cd\times\Nd}$, with ${\bf I}$ denoting the identity
  \footnote{Some authors add a so-called Cesaro averaging routine on top of this QW model \cite{aharonov2001quantum,marquezino2008mixing}.
  Our results can explicitly capture this and similar extensions via local stochastic maps, see Supplemental Material in appendix \ref{appendix}.}.
The graph locality is imposed by $\braket{c',v'|{\bf M_k}|c,v} = 0$ if $(v,v') \notin {\cal E}$.
To complete the setup, an initial distribution $p_0$ over $\Nd$ is mapped onto the lifted nodes (or sublevels) by $F: p_0 \mapsto \rho=\sum_{v \in \Nd} p_0(v)\q{c_v, v}\bra{c_v,v},$ thus associating some fixed initial coin state $c_v$ to each $v.$
The object of interest is the distribution $p_t$ over $\Nd$, the main nodes or levels \footnote{
Standard literature like \cite{chen1999lifting} defines LMCs with \emph{joint} distribution over $\Cd \times \Nd$ as object of interest, without initialization map $F$.
This does not affect our QW results, and in fact it implies no significant difference in general, see \cite{apers2017lifting}.}, obtained with the partial trace as $p_t =\textrm{diag}( \mathrm{trace}_\Cd(\rho_t))$.

Similarly, a LMC follows the dynamics
$\; \joint_{t+1} = \P \, \joint_t, \;$
where $\joint_t$ is a vector representing the probability distribution over $\Cd \times \Nd,$ and $\P$ is a stochastic matrix expressing the jump probabilities among sublevels. Namely, denoting by $p=\cl{v}$ and $\joint=\cl{(c,v)}$ the distributions with probability 1 of being on $v$ and on $(c,v),$ respectively, $\clt{(c',v')}\,\P\, \cl{(c,v)}$ is the transition probability from $(c,v)$ to $(c',v')$. Graph locality imposes $\clt{(c',v')}\,\P\,\cl{(c,v)} = 0$ if $(v,v') \notin {\cal E}$. Initial lifted nodes are assigned by $F: p_0 \mapsto \joint_0=\sum_vp_0(v)\cl{(c_v, v)}$. The distribution of interest is obtained by marginalizing over $\Cd$, thus $p_t(v)=\sum_{c \in \Cd} \joint_t(c,v)$ for all $v \in \Nd$.

Clearly, a LMC is a particular QW where populations evolve without coherences, i.e., where $\rho_t$ remains diagonal at all times and ${\bf M_k} = \sqrt{\clt{(c',v')}\,\P\,\cl{(c,v)}}\; \ket{c',v'}\bra{c,v}$, with index $k$ running over all nonzero elements of $\P$. The key to our main result will be to observe how, conversely, any QW can be simulated by some LMC (with possibly higher-dimensional coin). In other words, the non-Markovian evolution of $p_t$ under a QW can be described as a classical Markovian evolution of sublevel populations.

We focus on comparing the mixing behavior induced by QWs and LMCs.
A QW or LMC mixes to some distribution $\pilim$ over $\Nd$ if for any initial state $F(p_0)$ the induced distribution $p_t$ converges to $\pilim$. The mixing time $\tau(\epsilon)$, for any $0<\epsilon<1$, is the time required to get $\epsilon$-close to the limit distribution in total variation distance, i.e., the smallest time such that $\frac{1}{2}\sum_{v \in \Nd}|p_t(v)-\pilim (v)|\leq\epsilon$ for all $t\geq \tau(\epsilon)$ and all $p_0$. A standard ``stabilizing'' requirement for a process that converges to $\pilim$ is that $p_0=\pilim$ should imply $p_t=\pilim$ at all times. This holds automatically for the time-invariant $\P$ considered by the LMC framework. The QW framework allows the $\bf M_k$ to depend on time, but in standard constructions only through the measurement mechanism, like making $q$ time-dependent in the cycle example (see \cite{kendon2007decoherence} for a review). Such QWs too preserve $p_t=\pilim$ at all times when $p_0=\pilim$. We call this property $\pilim$-invariance \footnote{Note that this condition involves both the channel ${\bf M_k}$ and the initialization $F$.} and we will come back to its significance. Our first result shows that the mixing performance of such QW can be closely matched by a LMC.
\begin{thm}
Given a $\pilim$-invariant QW with mixing time $\bar{\tau}(\epsilon_0)$ for some $\epsilon_0\leq 1/4$, we can construct an LMC that has mixing time $\,\tau(\epsilon) \, / \, \bar{\tau}(\epsilon_0) \, \leq  \,\lceil \log(1/\epsilon)\,/\,\log(1/(2\epsilon_0))\rceil\,$ for all $\epsilon>0$.
\end{thm}

Mathematical details for all our results are available in the Supplemental Material, appendix \ref{appendix}.
The main idea in proving Thm.1 is to simulate the QW over the time interval $[0,\bar{\tau}(\epsilon_0)]$ using a LMC.
Indeed, as shown for unitary evolutions in \cite{aaronson2005quantum}, the probability distribution in the \emph{fixed measurement basis} associated to the nodes is not subject to the no-go results for general local hidden variables theories.
We extend this result to $p_t$ induced by an arbitrary QW that starts from a given node $v \in \Nd$. Following $p_t$ step by step, one builds a sequence of stochastic matrices $\P_1^{(v)},\P_2^{(v)},...,\P_t^{(v)}$ acting on $\Nd$ only, satisfying graph locality, and such that $p_t = \P_t^{(v)} p_{t-1}$ when starting on $v$.
The max-flow min-cut theorem from graph theory ensures that such construction always exists.
It can be traced back to a property that holds for LMCs, QWs, and more general local stochastic processes independently of the underlying physical mechanism: a node cannot contain more population at time $t+1$, than the population at time $t$ on itself and on its neighbors \cite{aharonov2001quantum}.
We thus simulate the QW with a classical process whose jump probabilities depend on time and on the starting node $v$.
To obtain a simulation with a (time-independent) LMC, at least for finite time horizon $t \leq \bar{\tau}(\epsilon_0)$, we encode these dependencies into the coin.
This follows the same spirit as adding registers in the clock Hamiltonians by Feynman and Kitaev \cite{feynman1982simulating,kitaev2002classical}.
Explicitly, we let current time $l$ and initial node $v$ act as a coin degree of freedom $c'=(v,l)$, which conditionally selects the proper transition matrix $\P_l^{(v)}$, see Fig.~\ref{fig:bridge}.
The resulting LMC describes a distribution over 
$\Cd' \times \Nd \equiv \big( \Nd \times \{0,1,...,T=\bar{\tau}(\epsilon_0)\}\big) \times \Nd \; ,$
with associated stochastic transition matrix
\[ \P \equiv \sum_{v \in \Nd} \cl{v}\clt{v} \otimes 
   \left( \sum_{t=0}^{T-1} \cl{t\text{+1}}\clt{t} \otimes \P_t^{(v)} + \cl{T}\clt{T} \otimes {\bf I}_\Nd \right)  \; \]
and initial assignment $F : \cl{v} \mapsto (\cl{v} \otimes \cl{0}) \otimes \cl{v}$. Finally, we apply an amplification technique that is exploited in randomized algorithms: the action of $\P$ on $\cl{T}$ is modified to have $\P \, \cl{(v_0,T,v)} = \cl{(v,0,v)} = F(v)$, so that the $T$ first steps are repeated iteratively. Thanks to $\pilim$-invariance of the QW that was used to generate $\P$, the resulting LMC will contract towards $\pilim$ at the announced exponential rate for all $t \geq T$.\\

\begin{figure}[!t]
 \center
  \includegraphics[width=88mm,trim=3.9cm 5cm 5cm 8cm,clip=true]{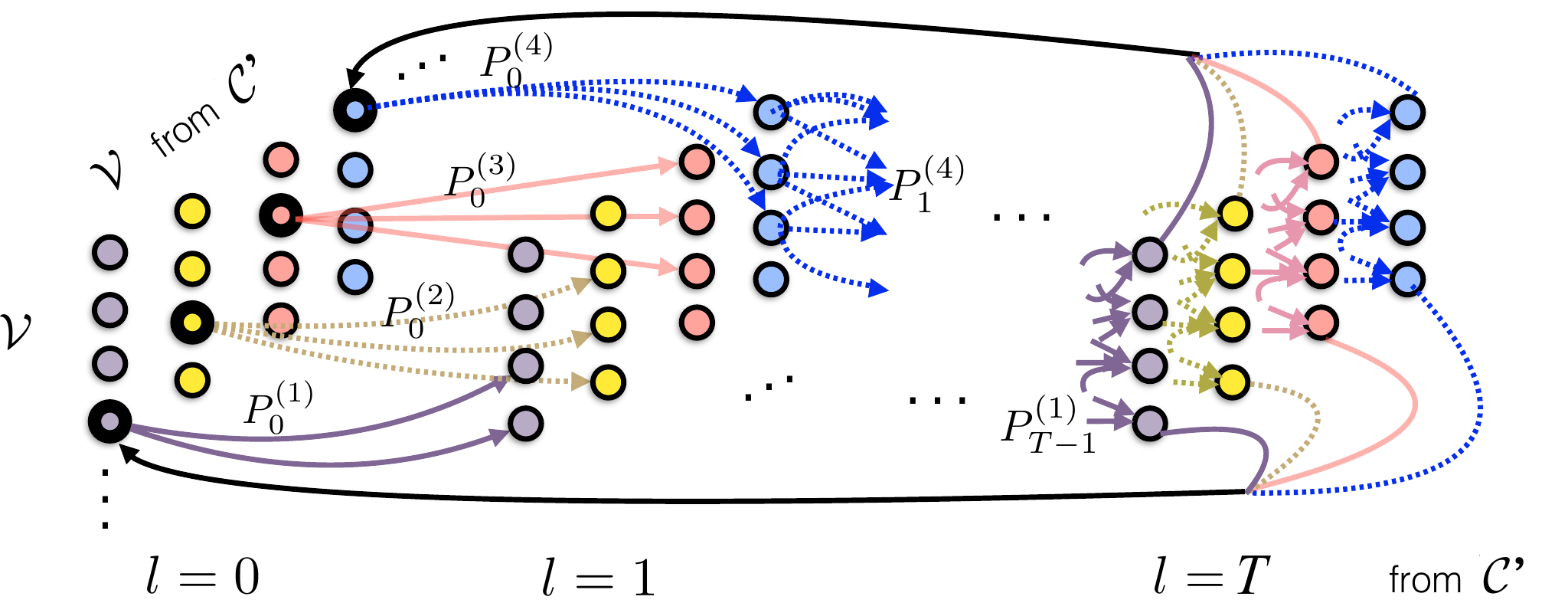}
\caption{Sketch of the LMC construction proving Thm.1, and implying that QWs feature the same conductance bound on fast mixing as LMCs. A given graph with nodes $\Nd$ (vertical axis) is lifted with coin space $\Cd'$ comprising both an initial node index (depth) and a time index (horizontal). LMC transitions are constructed between nodes of this lifted graph. (edges shown partially to avoid clutter)}\label{fig:bridge}
\end{figure}

Beyond the comparison with LMCs, this construction implies a general bound on the mixing performance of $\pilim$-invariant QWs. This tightens and generalizes the bounds of \cite{aharonov2001quantum,temme2010chi}, which are restricted to generating uniform $\pilim$ with unital quantum channels. The bound involves a function of graph topology and target distribution only, meant to capture the bottlenecks that slow down mixing, called the graph conductance $\Phi_{\pilim}$. Specifically, partitioning $\Nd$ into two subsets $\X$ and $\X^c$, consider that all the stationary population on $\X^c$ is lost; the conductance counts which fraction of the remaining population $\pilim(\X)=\sum_{v\in\X} \pilim(v)$ jumps back to $\X^c$ in one step. More precisely, if $\P$ on $\Nd$ has a stationary distribution $\pilim$, then
      \[ \Phi_{\pilim}(\P) =  \min_{\X: 0<\pilim(\X) \leq \frac{1}{2}} \, \Big( \, \sum_{v\in \X,v'\in \X^c} (\clt{v'} \P \cl{v})\pilim(v) \, \Big) \, / \, \pilim(\X) \; .
      \]
The maximal $\Phi_{\pilim}(\P)$ over all Markov chains that keep $\pilim$ invariant on a given graph, is the graph conductance $\Phi_{\pilim}$.

The estimate $1/\Phi_{\pilim}$ is a well-known lower bound on the mixing time of any classical Markov chain, and it carries over to the convergence of $\joint_t$ in associated LMCs \cite{chen1999lifting}. Conversely, \cite{chen1999lifting} establishes a construction of LMCs that essentially saturate this bound; it however requires to solve a hard multi-commodity flow problem over the entire graph. A novel observation, obtained essentially by fully exploiting the triangle inequality while computing the marginal probabilities, is that the bound keeps holding when taking the marginal $p_t$ over sublevels (i.e.,~over coin values) of a $\pilim$-invariant LMC. Combining this with Theorem 1 provides a tight bound for the ultimately achievable mixing time of QWs.
\begin{thm} Any $\pilim$-invariant QW has a mixing time $\tau(1/4) \geq 1/(4\Phi_{\pilim})$, and there exists such a QW that has a mixing time
  $\tau(\epsilon) \leq O(\,\log(1/\min_k\pilim_k) \log(1/\epsilon){\; /\; \Phi_{\pilim}}\,)$ for all $\epsilon>0$.
\end{thm}\vspace{2mm}
Besides mixing, the LMC construction has relevance for other tasks, enabling for instance to effectively simulate quantum transport with finite classical resources.

\noindent \emph{On efficient design of fast mixing QWs and LMCs. --} Fast mixing LMCs can often be built significantly more simply than with the general construction of Thm.1, by mirroring the structure of corresponding QWs. Accelerated mixing with QWs has been mostly demonstrated for graphs with strong symmetries, more specifically lattices \cite{ambainis2001one,aharonov2001quantum,moore2002quantum,richter2007quantum,marquezino2008mixing,marquezino2010mixing}. Similarly to the QW on the circle above, these examples use coin values to encode the lattice generators among which the walker can select its next move. 

Remarkably, the same structure is found in a proposal for designing fast mixing LMCs \cite{diaconis2013spectral}. For a $d$-dimensional square lattice of size $M$, the coin features $2d$ values of type $\pm_{k}$, with $k \in \{1,2,...,d\}$ indicating the axis and $\pm$ the direction of conditional motion among the nodes. At each step, the coin has a probability $\alpha = 1/(2dM)$ to switch to each of the other coin values, thus retaining a high probability $1-(2d-1)/(2dM)$ to stay with the same generator.
This dynamics precisely corresponds to a QW with diagonally dominant coin update $\bf C$ that is projectively measured at each step, as in Eq.\eqref{eq:Uex} with $q=1$. For fixed dimension $d$, it also provides the same order of speedup as a QW with $q \ll 1$ \cite{richter2007quantum}, and as the best possible QW according to Theorem 2, namely linear in $M$. Indeed, by counting the probabilities of applying, to each lattice dimension consecutively, the sequence of steps that lead to fast mixing on the cycle, one obtains the following (possibly loose) bound \footnote{This bound was conjectured to hold more generally for all Abelian Cayley graphs in \cite{diaconis2013spectral}, and the authors provide a concrete proof only for the case of the circle.}.
\begin{thm}
The just described LMC on $\integer_N^d$ has a mixing time
   $\; \tau(\epsilon) 
        \leq O(M \; d^2\log(d)\log(1/\epsilon)) \; .
    $
\end{thm}
\noindent Thus, QRW and LMC have the same order of mixing time; the same structure; and they require the same graph knowledge for tuning ($\alpha$ and/or $q$), namely the time $O(M)$ at which mixing will be considered accomplished.\\

In summary, we clarify that QWs on a graph induce non-Markovian local processes whose mixing behavior can be simulated by LMCs (Thm.1), and that this has several implications for searching a quantum speedup in mixing processes. The construction of Thm.1 can in fact be extended to abstract local stochastic dynamics (see Supplemental Material in appendix \ref{appendix}) beyond the QW model. As a consequence, the hierarchy LMCs $\subseteq$ QWs $\subseteq$ \{general local processes\} collapses regarding mixing speed, not only in terms of ultimate speedup achievable on general graphs (Thm.2), but also in terms of paradigmatic cases for which efficient mixing designs are known (lattices, Thm.3). In this light, a mixing speedup with respect to Markov chains is not diagnostic of underlying quantum dynamics, but potentially just of a memory effect. This prompts the question whether there is room for a ``quantum advantage'' at all in QW mixing. 

Besides establishing that there is no advantage in terms of best achievable mixing time, our analysis also suggests why this is not the end of the story. While the property of $\pilim$-invariance holds and stabilizes the system in typical QW proposals, it does not hold in some mixing-related applications, like simulated annealing. This distinction may be important as, without $\pilim$-invariance, the conductance bound of Thm.2 could be broken significantly \cite{apers2017lifting}.
As another memory-related aspect, in Eq.~\eqref{eq:Uex} on the cycle, taking $\alpha=1/2$ leads to fast QWs, while the corresponding ``projectively measured'' LMC boils down to the quadratically slower standard random walk. This shows that coherences can play a beneficial role, and could guide future research towards designing simple yet fast mixing QWs on graphs for which, unlike on lattices, LMCs of simple design do not meet the conductance bound yet. Furthermore, the QW of Eq.~\eqref{eq:Uex}, taking $\alpha=1/2$ and $q=1/N$, turns out to efficiently mix over the $t$ nodes closest to its starting node, \emph{for any number of iterations $t<N$} \cite{ambainis2001one}. Such multiscale mixing cannot be achieved with the LMC of Eq.~\eqref{eq:Pex}, where tuning $\alpha = 1/N$ to have good mixing at $t=N$ implies almost deterministic motion for $t\ll N$. This feature could point to efficient quantum algorithms addressing tasks related to mixing, yet not directly reducible to it.

The authors want to thank Giuseppe Vallone and Lorenza Viola for valuable suggestions and comments on earlier versions of the manuscript.

\bibliographystyle{ieeetr}
\bibliography{biblio}

\appendix
\section{Supplemental Material} \label{appendix}
The objective of the paper is the comparison of Quantum Walks (QWs) and Lifted Markov Chains (LMCs). However, the main results can be extended to a more general setting that includes both QWs and LMCs, namely, local stochastic processes preserving the target distribution.
An example of such a generalized setting would be Cesaro averaging, i.e., to consider as output distribution the uniform time average of the evolution generated by a QW or LMC.
Here we shall provide detailed proofs of our results directly in this generalized setting. The main ideas remain the same as for the particular case of QWs.

The supplemental material is organized as follows. We start with some notation and defining the generalized class of processes that will be studied. We then explicitly show how QWs fall under this setting. The next three sections are respectively devoted to a detailed proof, with all mathematical details worked out, of each of the three theorems of the main paper.\\

\emph{Notation --} We first recall some notation that will be used throughout these notes. Let $\Gr$ be a graph with node set $\s{V}$ and edge set $\s{E}\subseteq \s{V} \times \s{V}$.
By convention, we include in the edge set all $(v,v)$, $v \in \s{V}$.
We define the in-neighborhood or simply neighborhood of a set $\s{X}\subseteq\s{V}$ as $\s{B}(\s{X}) = \{v\in\s{V}\backslash\s{X}:(v,v')\in\s{E} \text{ for some } v'\in\s{X}\}.$ Note that we will keep, throughout the report, this notation such that ``probability mass flows from $v$ to $v'\;$''. We create a \textit{lifted graph} by expanding the node set to $\s{C} \times \s{V}$, for some finite set $\s{C}$, so the nodes for the lifted graph are pairs $(c,v)$, with $\,c\in\s{C},\, v\in\s{V},$ and we let the edge set be a subset of $\big\{\,\big((c,v),(c',v')\big)\,|\,(v,v')\in\s{E}\big\}$.
We associate the Hilbert space $\Hi_{\s{C}\times \s{V}} = \sp\{\ket{c,v}|c\in\s{C},v\in\s{V} \}$ to this graph and call $\dos(\Hi_{\s{C}\times\s{V}})$ the set of density operators over $\Hi_{\s{C}\times\s{V}}$, that is, positive semidefinite Hermitian operators of trace one.
For any subset $\s{X}\subseteq\s{V}$ and $\dens \in\dos(\Hi_{\s{C}\times \s{V}})$, we define $\pg{\s{X}}{\dens} = \tr(\o{\Pi}_{\s{X}}\dens)$, where $\o{\Pi}_{\s{X}} = \id \otimes \sum_{v\in\s{X}}\ket{v}\bra{v}$ is the projector onto the subspace associated to the subset of nodes $\s{X}$ of the original graph $\Gr$. More generally, we will use the standard notation $\pg{E}{p}$ to denote the probability of some event $E$ according to the probability measure $p$; occasionally we also use the notation $\pg{E}{p}=p(E)$. We will also denote, as in the main paper, by $\cl{v}$ the probability (column) vector with all weight on the node $v \in \s{V}$, and by $\cl{v}^\dagger$ the dual classical (row) vector. Using the tensor product $\cl{c} \otimes \cl{v} = \cl{(c,v)}$ also known as the Kronecker product of vectors, we get the probability vector with all weight on the single event $(c,v) \in \s{C} \times \s{V}$.

\subsection{Local stochastic processes}

In this section we introduce the concept of local stochastic dynamics, and we show how QWs and LMCs fall under this general framework. A stochastic map $\gc$ over $\s{V}$ is function that maps a probability distribution $\p_0$ to another probability distribution $\p_1$; it is linear and preserves both the positivity and the sum of the components of $p_0$. The general stochastic processes which we consider are a family of stochastic linear maps $\gc_t$, indexed by time $t\in\natural$, and which map an initial probability distribution $\p_0$ over $\s{V}$ to a probability distribution $\p_t=\gc_t[\p_0]$ over $\s{V}$ at each time $t$.
We say that the family $\gc_t$ is \textit{local} with respect to a graph $\Gr$ with nodes $\s{V}$ if and only if \cite{aharonov2001quantum}:
\begin{equation}\label{eq:rtuiop}
      \text{For all } \s{X}\subseteq \s{V}, \p_0, t> 0,\; 
          {\text{ it holds that }} \;\;
              \pg{\s{X}}{\p_{t+1}} \leq \pg{\s{X}}{\p_t} + \pg{\s{B}(\s{X})}{\p_t}.
\end{equation}
This formula expresses the intuitive statement from the main paper, that a node $\s{X} = \{v\}$ cannot contain more population at time $t+1$, than the population at time $t$ on itself and on its neighbors $\s{B}(v)$.
  
The family $\gc_t$ is \textit{invariant} with respect to a distribution $\pilim$, or short $\pilim$-invariant, if and only if $\gc_t[\pilim] = \pilim$, $\forall t\in\natural$. This expresses that $\p_t = \pilim$ for all $t\geq 0$ when $\p_0=\pilim$, and it ensures that the process \emph{stabilizes} $\pilim$ at all times.

\subsubsection{Quantum channels as $\pilim$-invariant local stochastic processes}

We will now show how such a family of abstract processes $\gc_t$ explicitly covers the specific case of $\p_t$ induced by QWs.
Thereto, let $\qc:\dos(\Hi_{\s{C}\times\s{V}})\to\dos(\Hi_{\s{C}\times\s{V}})$ be a completely positive trace-preserving (CPTP) map representing a QW, defined by
  \[ \qc[\dens] = \sum_k\o{M_k} \dens \o{M_k}^\dag, \]
together with a linear initialization map $\o{F}:p_0\mapsto \sum_{v\in \s{V}} \p_0(v) \ket{c_v,v}\bra{c_v,v}$.
For any starting condition $p_0 = \cl{v}$, we can compute the distribution $\p_t$ induced by the QW as the diagonal of the partial trace over $\s{C}$ of $\qc^t[\o{F}[\p_0]]$. Thanks to linearity of all these steps, the computation of the resulting evolutions $\p_0,\p_1,...$ can be described by a family of linear maps $\gc_t$ such that $\p_t = \gc_t[\p_0]$, for general $\p_0$ too.
Of course they preserve positivity and total probability, so they are stochastic; and if a target distribution $\pilim$ is invariant under $\qc$, then  obviously it is invariant too under the family of induced $\gc_t$. In the following lemma we prove that if $\qc$ is local, in the sense that the $\o{M_k}$ have zero entries where nodes are not connected in $\Gr$, then so is $\gc_t$ in the sense of Eq.\eqref{eq:rtuiop}.

\begin{lem} \label{lem:local}
Let $\qc$ be a quantum channel.
The following statements are equivalent:
\begin{itemize}
  \item[(a)] For all $c,c'\in\s{C},\;v,v'\in\s{V}$, it holds: if $(v,v') \notin \s{E} \text{ then }\Braket{c',v'|\o{M}_l|c,v} = 0$ $\forall l$.
  \item[(b)] For all $\s{X}\subseteq \s{V}$ and $\dens \in \dos(\Hi_{\s{C}\times\s{V}}),\; {\text{ it holds that }} \;\; \pg{\s{X}}{\qc[\rho]} \leq \pg{\s{X}}{\rho} + \pg{\s{B}(\s{X})}{\rho}.$
\end{itemize}
\end{lem}

\begin{proof}
``\textit{(a)} $\Rightarrow$ \textit{(b)}": We will show that the inequality in \textit{(b)} holds for a one-dimensional projection $\rho=\ket{\psi}\bra{\psi},$ $\ket{\psi}\in\Hi_{\s{C}\times\s{V}}$; due to linearity of the involved operators in $\rho$, the inequality must then necessarily hold also for all density operators $\rho\in\dos(\Hi_{\s{C}\times\s{V}})$, being convex combinations of projections.
We can write
\begin{align*}
  \pg{\s{X}}{\qc[\ket{\psi}\bra{\psi}]}
    &= \sum_{c'\in\s{C},v'\in\s{X}} \braket{c',v'|\qc[\ket{\psi}\bra{\psi}]|c',v'} \\
    &= \sum_{c'\in\s{C},v'\in\s{X}}\sum_l\braket{c',v'|\o{M}_l|\psi}\braket{\psi|\o{M}_l^\dag|c',v'} \\
    &= \sum_{c'\in\s{C},v'\in\s{X}}\sum_l|\Braket{c',v'|\o{M}_l|\psi}|^2. 
 \end{align*}
Since we assume that \textit{(a)} holds, we have that $\forall l$, $\; \Braket{c',v'|\o{M}_l|c,v}=0$ for $v'\in \s{X}$ and $v\in \s{X}^c\backslash \s{B}(\s{X})$, where $\s{X}^c=\s{V}\backslash \s{X}$. 
If we now write $\ket{\psi}=\ket{\psi_\s{X}}+\ket{\psi_{\s{B}_\s{X}}}+\ket{\psi_{\s{X}^c\backslash \s{B}_\s{X}}}$, where $\ket{\psi_{\s{Y}}} \equiv \o{\Pi}_{\s{Y}}\ket{\psi}$ for any $\s{Y}\subseteq \s{V}$, then \textit{(a)}, in particular, implies that $\forall c' \in \s{C}, v'\in\s{X}$ and all $l$, we have $\; \Braket{c',v'|\o{M}_l|\psi_{X^c\backslash B_X}}=0.$
Intuitively, this expresses that $\ket{\psi_{\s{X}^c\backslash \s{B}_{\s{X}}}}$ does not contribute to the probability of observing $\Pi_X$ after the action of $\qc$. Inserting $\Braket{c',v'|\o{M}_l|\psi_{X^c\backslash B_X}}=0$ into the above sum, we thus get:
\begin{eqnarray*}
  \pg{\s{X}}{\qc[\ket{\psi}\bra{\psi}]}
    &=& \pg{\s{X}}{\qc[(\ket{\psi_{\s{X}}}+\ket{\psi_{\s{B}(\s{X})}})(\bra{\psi_{\s{X}}} +\bra{\psi_{\s{B}(\s{X})}})]} \\
&&\!\!\!\!\!\!\!\text{\scriptsize (going from population on $\s{X} \cup \s{B}(\s{X})$ to population on $\s{X}$ after applying $\qc$)}\\    &\leq& 
\tr(\qc[(\ket{\psi_{\s{X}}}+\ket{\psi_{\s{B}(\s{X})}})(\bra{\psi_{\s{X}}} +\bra{\psi_{\s{B}(\s{X})}})])\\
&&\!\!\!\!\!\!\!\text{\scriptsize ($\qc$ is trace-preserving and $\ket{\psi_{\s{X}}}$ orthogonal to $\ket{\psi_{\s{B}(\s{X})}}$)}\\   &=& 
\tr((\ket{\psi_{\s{X}}}+\ket{\psi_{\s{B}(\s{X})}})(\bra{\psi_{\s{X}}} +\bra{\psi_{\s{B}(\s{X})}})) \\
&=& \tr(\ket{\psi_{\s{X}}}\bra{\psi_{\s{X}}})
              +\tr(\ket{\psi_{\s{B}_{\s{X}}}}\bra{\psi_{\s{B}_{\s{X}}}}) \\
    &=& \pg{\s{X}}{\ket{\psi}\bra{\psi}} + \pg{\s{B}(\s{X})}{\ket{\psi}\bra{\psi}}.
\end{eqnarray*}

"\textit{(b)} $\Rightarrow$ \textit{(a)}": Assume that \textit{(a)} does not hold.
Thus, there exists some $l$, some $c,c'\in\s{C}$, some $v'\in \s{X}$ and $v\in \s{X}^c\backslash \s{B}_{\s{X}}\;$ such that $\; \Braket{c',v'|\o{M}_l|c,v} \neq 0$.
If we now consider $\ket{\psi}=\ket{c,v}$, then $\pg{\s{X}}{\ket{c,v}\bra{c,v}} + \pg{\s{B}(\s{X})}{\ket{c,v}\bra{c,v}} = 0$, whereas $\pg{\s{X}}{\qc[\ket{c,v}\bra{c,v}] }\geq \pg{(c',v')}{\qc[\ket{c,v}\bra{c,v}]} = |\Braket{c',v'|\o{M}_l|c,v}|^2 >0$. So \textit{(b)} does not hold when \textit{(a)} does not; thus conversely, if \textit{(b)} holds then \textit{(a)} must hold too.
\end{proof}
In the light of the above lemma, a quantum channel is said to be {\em local} with respect to a reference lifted graph if $(a),$ or equivalently $(b),$ holds; and from (b) thus, the associated $\gc_t$ will be local in the sense of Eq.\eqref{eq:rtuiop} too.

\subsection{Proof of Theorem 1}\label{sec:II}

In the main paper, Theorem 1 essentially states that QW mixing can be simulated by an LMC, and the main steps of its proof are described for this setting.
Here we provide a formal proof for a more general statement: the mixing performance of any stochastic process that is local and invariant can be simulated using a suitably constructed local LMC.

\subsubsection{Simulability of stochastic linear maps}

In the following, we first show that the $p_t$ generated by a local stochastic map, starting {\em from any given initial distribution}, can always be simulated by a sequence of \emph{stochastic transition matrices} which \emph{each satisfy the graph locality.} This sequence will be in general dependent on the initial distribution.
The lemma and proof are a generalization of the result by \cite{aaronson2005quantum} from unitary evolution to abstract stochastic linear maps.

\begin{lem}[Local simulability] \label{lem:sim}
 If $\gc_t$ is local, then for every pair $(\p_0,t)$ with $t>0$ there exists a local stochastic matrix $\o{P}_t^{(\p_0)}$ such that $\p_t = \o{P}_t^{(\p_0)}\p_{t-1}$, where $\p_t = \gc_t[\p_0]$.
\end{lem}
\begin{proof}
 Call $y=\p_{t-1}$ and $z=\p_t$. In order to prove the above statement, it is convenient to resort to results concerning \emph{flows over capacitated networks} \cite{ford1956maximal}, and, in particular, consider the graph shown in Figure \ref{fig:stoch-bridge}, where each edge is assigned a corresponding weight, or capacity.
 The network consists of a source node $s$, a sink node $r$, and two copies $\s{W}$ and $\s{W}'$ of the set of node states $\s{V}$.
Node $s$ is connected with capacity $y(v)$ to any node $v\in \s{W}$; any node $v\in \s{W}$ is connected with capacity 1 to any node $v'\in \s{W}'$ iff $(v,v')\in \s{E}$, else the nodes are not connected; and any node $v'\in \s{W}'$ is connected with capacity $z(v')$ to node $r$. The capacities $y(v)$ and $z(v')$, respectively from $s$ and to $r$, thus reflect the probability distributions to be mapped.
The key observation is the following: if this network can route a steady flow of value 1 from node $s$ to node $r$, then the fraction from $v \in \s{W}$ that is routed towards $v'\in \s{W}'$ directly defines the entry $e^\dag_{v'}\o{P}^{(\p_0)}_te_{v}$ that we need, and also denoted $\o{P}^{(\p_0)}_t(v',v)$.
Indeed, to route a flow of value 1 the edges from $r$ to $\s{W}$ will have to be used to their full capacities $y(v)$, such that the flow through the edges from $\s{W'}$ to $s$ becomes $z(v') = \sum_{v\in\s{V}} \o{P}_t^{(\p_0)}(v',v)\,y(v)$; so we would have $\o{P}_t^{(\p_0)}y = z$ as claimed.
 
    \begin{figure}[htb]
     \centering
     \def\svgwidth{.8\columnwidth}
     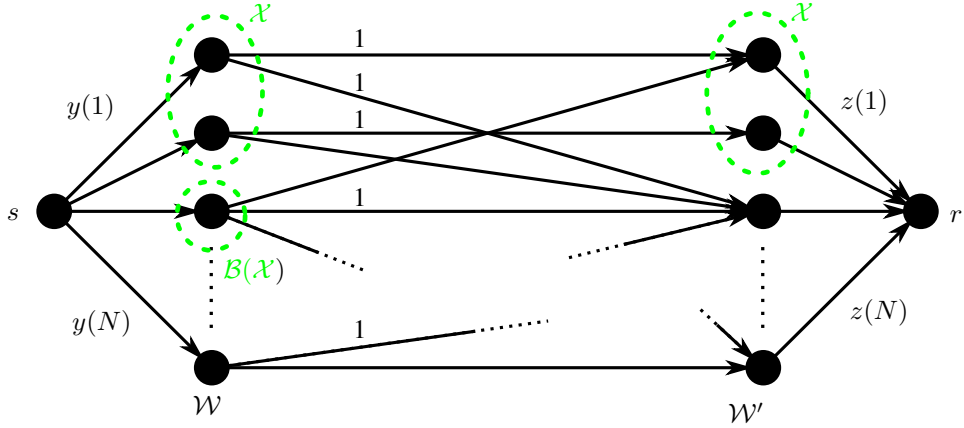
     \caption{Capacitated network construction used in Lemma \ref{lem:sim}.}
     \label{fig:stoch-bridge}
    \end{figure}
 
The max-flow min-cut theorem \cite{ford1956maximal} states that the maximum steady flow which can be routed from node $s$ to node $r$ is equal to the minimum cut value of the graph, where a cut value is the sum of the capacities of a set of edges that disconnects $s$ from $r$.
 It is clear that cutting all edges arriving at $r$ disconnects the graph, with a cut value of $1$, whereas cutting any middle edge between $\s{W}$ and $\s{W}'$ gives a cut value $\geq 1$.
So the minimum cut should not include any of these ``middle'' edges, and it must be some combination of edges starting on $s$ or arriving at $r$. 
Assume that we know the optimal cut, and let $\s{X} \subseteq \s{W}'$ such that the cut involves the edges from the {\em complement} of $\s{X} \subseteq \s{W}'$ to $r$. To block any flow from $s$ to $r$ while keeping all middle edges, we must then cut the edges from $s$ to all the $l \in \s{W}$ which have an edge to $\s{X}$.
This corresponds to all $l \in \s{X} \cup \s{B}(\s{X})$. The value of this cut is thus
   \[ 1 - \pg{\s{X}}{z} + \pg{\s{X}}{y} + \pg{\s{B}(\s{X})}{y}. \]
Recalling that $y=\p_{t}$ and $z=\p_{t+1}$, locality imposes
   \[ \pg{\s{X}}{z} \leq \pg{\s{X}}{y} + \pg{\s{B}(\s{X})}{y}, \]
from which it follows that the minimum value of the cut is $\geq 1$. This minimum is attained (among others) with cutting all edges arriving at $r$, i.e., with $\s{X}$ the empty set. Hence, the minimum cut value is 1 and a $\o{P}^{(\p_0)}_t$ solution to our problem exists.\end{proof}

\subsubsection{Amplification lemma}

Lemma \ref{lem:sim} is instrumental in proving Thm.1 of the main paper for a finite time frame, by simulating the QW up to some given time. The following will be instrumental to prove the theorem for arbitrary time, showing that a \emph{finite-memory} process is sufficient to extend this mixing performance to arbitrarily small $\epsilon>0$. In particular, we now show that, given an evolution map that mixes up to a certain total variation distance, we can iterate this map in order to mix to arbitrarily small distance, a process informally known as {\em amplification}.

\begin{lem}[Amplification lemma] \label{lem:amp}
Assume that $\gc_t$ is a family of stochastic linear maps that mixes to an invariant $\pilim,$ and admits a mixing time $\overline{\tau}(\epsilon)$ for all $\epsilon>0$.
 Then for any $\epsilon_0 < 1/2$, its amplified version defined as
   \[ \widetilde{\gc}_t = \gc_{t\,\mathrm{mod}\,T} \left(\gc_T\right)^{\lfloor t/T\rfloor}, \]
 with $T = \overline{\tau}(\epsilon_0)$, has a mixing time $\tau(\epsilon) \leq \overline{\tau}(\epsilon_0) \cdot \lceil \,\log(1/\epsilon)\,/\,\log(1/(2\epsilon_0))\,\rceil$ for all $\epsilon>0$.
\end{lem}
\begin{proof}
We will thus check that at any time $t \geq T \cdot \lceil \,\log(1/\epsilon)\,/\,\log(1/(2\epsilon_0))\,\rceil = \kappa \cdot T$, $\kappa\in\mathbb N$, the total variation distance to $\pilim$ is lower than $\epsilon$. The proof uses invariance of $\pilim$ under $\gc_t$ to transform $\gc_t[\p]-\pilim$ into $\gc_t[\p-\pilim]$.

For $t=\kappa\cdot T$, we get
    \begin{eqnarray*}
      \max_\p \| \widetilde \gc_{\kappa\cdot T}[\p] - \pilim \|_{TV}
        &=& \max_\p \| \left(\gc_T\right)^\kappa[\p] - \pilim \|_{TV}
        \;\leq\; \max_{\p,\p'} \| \left(\gc_T\right)^\kappa[\p] - \left(\gc_T\right)^\kappa[\p'] \|_{TV}    \\  
&&\!\!\!\!\!\!\! \text{\scriptsize (see justification below)}  \\
    &\leq& \left( \max_{\p,\p'} \| \gc_{T}[\p] - \gc_{T}[\p'] \|_{TV} \right)^\kappa \\
        &\leq& \left( 2 \max_\p \| \gc_{T}[\p] - \pilim \|_{TV} \right)^\kappa \; \leq (2\epsilon_0)^\kappa \; \leq \epsilon \; .
    \end{eqnarray*}
On the last line we have used that $\max_\p \| \o{\Psi}_{T}[\p] - \pilim \|_{TV} \leq \epsilon_0$, and that $(2\epsilon_0)^\kappa \leq \epsilon$ as soon as $\kappa \geq \log(1/\epsilon)/\log(1/(2\epsilon_0))$; from first to second line we have used the submultiplicativity of the total variation norm under any linear map in the form stated in \cite{levin2017markov}.

For $t = t' + \kappa\cdot T$ with any $t'>0$, we know that 
    \[ \| \gc_t[\p] - \pilim \|_{TV}
            = \|  \gc_t[\p - \pilim] \|_{TV}
            = \| \gc_{t'}[\gc_{\kappa\cdot T}[\p-\pilim]] \|_{TV}
                      \leq \| \gc_{\kappa\cdot T}[\p-\pilim] \|_{TV}, \; \text{ for all } \p, \]
thanks to contractivity of the 1-norm under stochastic maps. So finally we find that, for arbitrary $t\geq 0$,
  \[ 
      \max_\p \| \widetilde \gc_t[\p] - \pilim \|_{TV}
        \leq \max_\p \| \widetilde \gc_{\lfloor t/T\rfloor \cdot T}[\p] - \pilim \|_{TV}
        \leq \epsilon,
  \]
if $t \geq \overline{\tau}(\epsilon_0) \cdot \lceil \,\log(1/\epsilon)\,/\,\log(1/(2\epsilon_0))\,\rceil$.
\end{proof}

\subsubsection{Proof of main theorem} \label{sec:thm}
We can now finalize the proof of a generalized form of Theorem 1 of the main paper for a general stochastic process associated to a family of linear maps $\gc_t$ over the node set $\s{V}$ that are local, and that leave the target distribution $\pilim$ invariant.
\begin{thm}[main paper Thm.1, generalized version] \label{thm:main}
 Let $\gc_t$ be a stochastic linear map that mixes to some distribution $\pilim$ with mixing time $\overline{\tau}(\epsilon)$, satisfying some locality constraint and leaving $\pilim$ invariant. Then for any $\epsilon_0 < 1/2$ we can construct an LMC that satisfies the same locality constraint and that mixes to $\pilim$ with a mixing time $\tau(\epsilon) = \overline{\tau}(\epsilon)$ for all $\epsilon\geq \epsilon_0$, and a mixing time
$\;\; \tau(\epsilon)
        \leq \overline{\tau}(\epsilon_0) \cdot \left\lceil\, \log(1/\epsilon)\,/\,\log(1/(2\epsilon_0)) \right\rceil$
 for all $\epsilon > 0$.
\end{thm}
\begin{proof}
The proof essentially combines the two previous lemmas, and for the rest it follows the construction from the main text. We shall first use Lemma \ref{lem:sim} to build a lifted Markov chain that simulates the dynamics of this channel up to time $t=\tau(\epsilon_0)$, and next apply the amplification lemma \ref{lem:amp} to prove exponential convergence for $\epsilon<\epsilon_0$ and thus $t>\tau(\epsilon_0)$.

\noindent $\bullet$ \textit{[First part: construction for $t \leq \tau(\epsilon_0)$]} Lemma \ref{lem:sim} tells us that for every initial state $\p_0$ and given time bound $T$, there exists a local \textit{stochastic bridge} $\{\o{P}_t^{\p_0}\}_{t=1,2,...,T}$ such that, for all $t \in (0,T]$, we have
    \[ \gc_t[\p_0] = \o{P}_t^{\p_0}\dots \o{P}_1^{\p_0} \p_0. \]
This allows us to construct the operator sets $\left\{\o{P}_t^{\cl{v}}\right\}_{t=1,2,...,T}$ for $v\in \s{V}$, where we recall that $\cl{v}$ is the elementary vector corresponding to node $v,$ representing the classical probability vector whose entire weight is on node $v$.
We will combine these bridges into a {\em single and time-invariant lifted Markov chain $\o{P}$}, mapping probability distributions over the extended set $\hat{\s{V}}=\s{C}'\times\s{V}=(\s{V}\times\{1,2,\dots,T\})\times\s{V}$, where $T=\overline{\tau}(\epsilon_0)$. Let $\cl{v_0}\otimes \cl{l}\otimes \cl{v}$ denote the probability (column) vector over $\hat{\s{V}}$ whose weight is centered on element $(v_0,l,v)\in \hat{\s{V}}$, and by $\clt{v_0}\otimes\clt{l}\otimes\clt{v}$ we denote the dual or adjoint (row) vector.
 Now we can define 
    \begin{equation} \label{eq:lift} 
      \o{P} \equiv \sum_{v\in\s{V}}
         \cl{v}\clt{v} \otimes \left( \sum_{t=0}^{T-1}\cl{t+1}\clt{t} \otimes P_t^{\cl{v}} + \cl{T}\clt{T} \otimes \id_{\s{V}}\right).
    \end{equation}
To complete the construction, the above evolution should be locally initialized according to the map $\o{F}$, which maps any probability vector $\p$ over $\s{V}$ to a probability vector $\o{F}[\p]$ over $\hat{\s{V}}$, defined as
    \begin{equation} \label{eq:init} 
      \o{F}[\p] = \sum_{v\in\s{V}} \p(v) \cl{v} \otimes \cl{0} \otimes \cl{v}.
    \end{equation}
 The probability vector $\o{F}[\p]$ is such that $\o{P}^t \o{F}[\p]$ induces the same marginal distribution on $\s{V}$ as $\gc_t[\p]$ for a fixed time frame:
    \[ f(\o{P}^t \o{F}[\p]) = \gc_t[\p] \quad \text{ for all } t \in [0,\overline{\tau}(\epsilon_0)]\;, \]
 where $f$ computes the marginal probability distribution induced by the lift on $\s{V}$, i.e.,
    \[ f \left( \sum_{v_0,l,v} p(v_0,l,v) \cl{v_0}\otimes\cl{l}\otimes\cl{v} \right) = \sum_v \left( \sum_{v_0,l} p(v_0,l,v)) \right) \cl{v}. \]
 As a consequence, initial states $\o{F}[\p]$, with an arbitrary $\p$ over $\Nd$, will mix on $\s{V}$ with the same mixing time $\overline{\tau}(\epsilon)$ as $\gc_t$ for any $\tau(\epsilon) \leq \tau(\epsilon_0)$, i.e., for any $\epsilon \geq \epsilon_0$. This proves the first part of the theorem.

\noindent $\bullet$ \textit{[Second part: modifying the construction towards $t > \tau(\epsilon_0)$]}
As the size of the lift transition matrix $\o{P}$ scales linearly with $\overline{\tau}(\epsilon_0)$, which in general is unbounded for $\epsilon_0 \to 0$, the above construction only makes sense for fixed $\epsilon_0$.
 Towards building a lift for arbitrary $\epsilon>0$, and thus prove the second part of the theorem, we invoke the amplification lemma \ref{lem:amp}.
 The lemma shows that, instead of the full process ${\gc}_t$, we can simulate the simpler one $\widetilde{\gc}_t$, defined as  
   \[ \widetilde{\gc}_t = \gc_{t\,\mathrm{mod}\,T} \left(\gc_T\right)^{\lfloor t/T\rfloor}, \]
and ensure a mixing time $\tau(\epsilon) \leq \tau(\epsilon_0) \cdot \lceil \log(1/\epsilon)\,/\,\log(1/(2\epsilon_0)) \rceil$.
It is not difficult to show that the evolution induced by $\widetilde{\gc}_t$ can in fact be simulated for an arbitrary number of steps, with a lift of fixed size.
To this aim, we modify the lift construction of the first part, namely replacing the unit probability of staying at ${(v_0,T,v)}$ by a unit probability to jump from ${(v_0,T,v)}$ to ${(v,0,v)}$. Explicitly, we thus adapt the lift as follows:
\begin{equation}\label{as:eq:modif}
\o{P} \equiv \sum_{v\in\s{V}} \sum_{t=0}^{T-1} \cl{v}\clt{v} \otimes \cl{t+1}\clt{t} \otimes P_t^{\cl{v}}
          + \sum_{v,v_0\in \s{V}} \cl{v}\clt{v_0} \otimes \cl{0}\clt{T} \otimes \cl{v}\clt{v}. \end{equation}
When associated to the same initialization map $\o{F}$ and marginalization $f$, this transition matrix gives exactly the same output distributions $\p_t$ over $\s{V}$ as the LMC of the first part, for all $t \leq T$. 
At $t=T$, in fact \eqref{as:eq:modif} takes the output $\p_T = f(\o{P}^T\o{F}[\p]) = \o{\Psi}_T[\p]$ of the LMC constructed in the first part, and reinitiates the walk with $F(\p_T)$ for the next steps. It follows that with \eqref{as:eq:modif} we have:
  \[ f(\o{P}^t \o{F}[\p]) = \widetilde{\gc}_t[\p],\quad \forall t\geq 0. \]
Lemma \ref{lem:amp} implies the conclusion about mixing time for all $\epsilon > 0$.
\end{proof}

The version reported in the main text is the second part of the above theorem for the special case of a stochastic process generated from a QW.

\subsection{Proof of conductance bound}

We have just shown that quantum channels, and stochastic linear maps in general, can be simulated by lifted Markov chains under appropriate conditions. Accordingly, we can prove a conductance bound for quantum channels and stochastic linear maps by building on a similar bound for LMCs that we provide in Lemma \ref{lem:lift-mixing} below.
It is a generalization of the bound formulated in for instance \cite{chen1999lifting}, to the setting where the Markov chain is initialized on a lifted space with some local map $\o{F}$ and where the convergence to a limit distribution only involves the \emph{marginal} over $\s{V}$.

Before going into the statement of the lemmas, let us formalize some concepts more rigorously with the notation of this Supplemental Material.
We say that an LMC $\o{P}$ with initialization map $\o{F}:\s{V}\to\s{C}\times \s{V}$ mixes to $\pilim$ with a \textit{mixing time} $\tau(\epsilon)$ for all $\epsilon>0$ if, for any $\p$ over $\s{V}$, the induced distribution of $\o{P}^t\o{F}[\p]$ over $\s{V}$ is $\epsilon$-close in total variation distance to $\pilim$ for all $t\geq \tau(\epsilon)$.
We will bound this mixing time using the \textit{conductance}, a quantity that we can associate to a general irreducible Markov chain $\o{P}$ on $\s{V}$. We recall that if $\o{P}$ has a stationary distribution $\pilim$, then its conductance $\Phi(\o{P})$ is defined as
      \[ \Phi(\o{P}) 
                          =  \min_{\s{X}\subseteq \s{V}: 0<\pilim(\s{X}) \leq \frac{1}{2}} \Phi_{\s{X}}(\o{P}), \quad
          \text{ with }\; \Phi_{\s{X}}(\o{P}) = \frac{\o{Q}_{\o{P}}(\s{X}^c,\s{X})}{\pi(\s{X})}, \]
where $\pilim(\s{X}) = \pg{\s{X}}{\pilim}$ and $\o{Q}_{\o{P}}(\s{X}^c,\s{X}) = \sum_{v\in \s{X},v'\in \s{X}^c}\o{P}(v',v)\pilim(v)$ is the ergodic flow from $\s{X}$ to its complement. Here we use as earlier the notation $\o{P}(v',v) = \clt{v'} \o{P} \cl{v}$.
We can also associate a conductance to a graph and distribution, without specifying an associated Markov chain. The \textit{graph conductance} $\Phi_{\pilim}$ with respect to some distribution $\pilim$ is defined as $\Phi_{\pilim} = \max_{\o{P}'}\Phi(\o{P}')$, where the maximization runs over all $\o{P'}$ satisfying the locality of the graph and $\o{P}'\pilim = \pilim$.

To any lifted Markov chain on a lifted graph, we can associate an \textit{induced chain} on the original graph, as introduced in \cite{aldous2002reversible}.
Thereto, let $\o{P}$ be an irreducible lifted Markov chain on the nodes of a lifted graph $\s{C} \times \s{V}$, having stationary distribution $\hat{\pilim}$.
The induced chain $\o{P}_{\s{V}}$ over $\s{V}$ is defined by
      \[  \o{P}_{\s{V}}(v',v)  = \sum_{c,c'\in\s{C}} \frac{\hat{\pilim}(c,v)}{\pilim(v)}\o{P}((c',v'),(c,v)), \]
where $\pilim$ represents the stationary distribution of $\o{P}_{\s{V}}$, defined by $\pilim(v) = \hat{\pilim}(\s{C} \times v)$.
This definition is motivated by obtaining matching ergodic flows $\o{Q}_{\o{P}_\s{V}}(v',v) = \o{Q}_{\o{P}}(\s{C}\times v',\s{C}\times v)$ and so for any subset $\s{X} \subseteq \s{V}$: 
        \begin{equation} \label{eq:marg-cond} 
          \Phi_{\s{X}}(\o{P}_{\s{V}}) = \Phi_{\s{C} \times \s{X}}(\o{P}).
        \end{equation}
This readily implies that $\Phi(\o{P}_{\s{V}}) \geq \Phi(\o{P})$. We also have that $\Phi_{\pilim} \geq \Phi(\o{P}_{\s{V}})$, with $\Phi_{\pilim}$ the graph conductance associated to $\pilim$ on the non-lifted graph. Indeed by definition $\o{P}_{\s{V}}$ obeys the graph locality and $\o{P}_{\s{V}}\pilim = \pilim$, i.e., it is an element of the set over which the graph conductance $\Phi_{\pilim}$ is maximized.

We next borrow standard techniques, as presented in for instance \cite{montenegro2006mathematical} and \cite{levin2017markov}, to prove two instrumental lemmas.
 \begin{lem} \label{lemma:cut}
   Consider an irreducible Markov chain $\o{P}$ over a set $\s{V}$, with unique stationary distribution $\pilim$. Then
      \[  \pg{\s{X}^c}{\o{P}^t\pilim_{\s{X}}} \leq t\, \Phi_{\s{X}}(\o{P}) \quad \text{for all } \s{X}\subseteq \s{V}, t\geq 0 \; , \]
   where $\pilim_{\s{X}}(v) = \pilim(v)/\pilim(\s{X})$, for $v\in\s{X}$, and zero elsewhere.
 \end{lem}
 \begin{proof}
   Note that
    \[ \pg{X^c}{\o{P}\pilim_{\s{X}}}
          = \sum_{v\in \s{X},v'\in \s{X}^c} \frac{\o{P}(v',v)\cdot \pilim(v)}{\pilim(\s{X})} 
                          = \Phi_{\s{X}}(\o{P}). \]
   We will first prove the following inequalities:
    \begin{equation} \label{eq:ineq} 
      \pg{\s{X}^c}{\o{P}\pilim_{\s{X}}} = \left\| \o{P}\pilim_{\s{X}}-\pilim_{\s{X}} \right\|_{TV}, \quad \text{and} \quad
      \pg{X^c}{\o{P}^t\pilim_{\s{X}}} \leq \left\| \o{P}^t\pilim_{\s{X}}-\pilim_{\s{X}} \right\|_{TV},
        \;\forall t\geq 0.
   \end{equation}
   To obtain the first inequality, we can use the equivalent definition of the total variation distance:
    \[ \left\| \o{P}\pilim_{\s{X}} - \pilim_{\s{X}} \right\|_{TV} 
        = \sum_{v\in V:(\o{P}\pilim_{\s{X}})(v) \geq \pilim_{\s{X}}(v)} (\o{P}\pilim_{\s{X}})(v) - \pilim_{\s{X}}(v). \]
   We then observe that on the set $\s{X}^c$, $\o{P}\pilim_{\s{X}}$ will be elementwise larger than or equal to $\pilim_{\s{X}}$, since the latter is zero on $\s{X}^c$; whereas on the set $\s{X}$, $\o{P} \pilim_{\s{X}}$ will be elementwise smaller than $\pilim_{\s{X}}$:
\[ (\o{P}\pilim_{\s{X}})(v') = \frac{\sum_{v\in \s{X}} \o{P}(v',v)\pilim(v)}{\pilim(\s{X})}
        \leq \frac{\sum_{v\in \s{V}}\o{P}(v',v) \pilim(v)}{\pilim(\s{X})} = \frac{\pilim(v')}{\pilim(\s{X})}=  \pilim_{\s{X}}(v') \quad \text{since $P\pilim=\pilim$}. \]
   Now we can rewrite $\sum_{v\in \s{V}:(\o{P}\pilim_{\s{X}})(v) \geq \pilim_{\s{X}}(v)} (\o{P}\pilim_{\s{X}})(v) - \pilim_{\s{X}}(v) = \sum_{v\in \s{X}^c} (\o{P}\pilim_{\s{X}})(v) = \pg{\s{X}^c}{\o{P}\pilim_{\s{X}}}$. To obtain the inequality in \eqref{eq:ineq} we expand the total variation norm:
    \begin{align*} 
      \| \o{P}^t\pilim_{\s{X}} - \pilim_{\s{X}} \|_{TV} 
        &= \frac{1}{2}\sum_{v\in \s{V}}|(\o{P}^t\pilim_{\s{X}})(v) - \pilim_{\s{X}}(v)|
      = \frac{1}{2}\pg{X^c}{\o{P}^t\pilim_{\s{X}}} + \frac{1}{2}\sum_{v\in \s{X}}|(\o{P}^t\pilim_{\s{X}})(v) 
                      - \pilim_{\s{X}}(v)| \\
      &\geq \frac{1}{2}\pg{\s{X}^c}{\o{P}^t\pilim_{\s{X}}} 
                      + \frac{1}{2}\left|\sum_{v\in \s{X}}(\o{P}^t\pilim_{\s{X}})(v) - \pilim_{\s{X}}(v)\right| \\
       &= \frac{1}{2}\pg{\s{X}^c}{\o{P}^t\pilim_{\s{X}}} + \frac{1}{2}\left|1-\pg{X^c}{\o{P}^t\pilim_{\s{X}}} - 1\right|
      = \pg{X^c}{\o{P}^t\pilim_{\s{X}}}.
    \end{align*}
Starting with this inequality, we apply the triangle inequality on $\o{P}^t\p - \p =(\o{P}^t\p - \o{P}^{t-1}\p ) + (\o{P}^{t-1}\p - \o{P}^{t-2}\p) + \dots + (\o{P}\p - \p)$; next we bound each term by $\left\| \o{P}\pilim_{\s{X}} - \pilim_{\s{X}} \right\|_{TV}$ thanks to contractivity, i.e., as at the end of the proof of Lemma \ref{lem:amp}, the fact that for arbitrary distributions $\p,\p'$ and a stochastic matrix $\o{P}$, we have
      $\;\; \|\o{P}\p - \o{P}\p' \|_{TV} \leq \| \p - \p' \|_{TV}\;\;$; and finally we apply the equality from \eqref{eq:ineq}. This yields:
      \[ \pg{X^c}{\o{P}^t\pilim_{\s{X}}}
                    \leq \left\| \o{P}^t\pilim_{\s{X}}-\pi_{\s{X}} \right\|_{TV}
            \leq t\left\| \o{P}\pilim_{\s{X}} - \pilim_{\s{X}} \right\|_{TV} = t\, \Phi_{\s{X}}(\o{P}). \]
  \end{proof}
  
  \begin{lem} \label{lem:lift-mixing}
    Consider an irreducible lifted Markov chain $\o{P}$ on a lifted state space $\s{C} \times \s{V}$, and call $\o{P}_{\s{V}}$ its induced chain on $\s{V}$. Then its mixing time on $\s{V}$ satisfies
      \[ \tau(1/4) \geq \frac{1}{4\Phi(\o{P}_{\s{V}})} \geq \frac{1}{4\Phi_{\pilim}}, \]
    where $\Phi_{\pilim}$ is the graph conductance associated to the corresponding limit distribution on $\s{V}$.
  \end{lem}
  \begin{proof}
    Let $f$ be the function computing the marginal distribution over $\s{V}$ from a distribution over $\s{C}\times \s{V}$. By applying the reverse triangle inequality, it is easily seen that
      \begin{align*}
        \left\| \o{P}^t\p - \hat{\pilim} \right\|_{TV}
          &= \frac{1}{2} \sum_{(c,v)\in \s{C}\times\s{V}}|(\o{P}^t\p)(c,v) - \hat{\pilim}(c,v)| \\
          &\geq \frac{1}{2} \sum_{v\in V}
                        \left| \sum_{c\in \s{C}}(\o{P}^t\p)(c,v) - \hat{\pilim}(c,v) \right|
          = \left\| f(\o{P}^t\p) - f(\hat{\pilim}) \right\|_{TV}.
       \end{align*}
    Now take a subset $\s{X}\subseteq \s{V}$ such that $\hat{\pilim}(\s{C}\times \s{X}) \leq 1/2$. Using this subset we define a second marginalization $g_{\s{X}}$ mapping distributions over the nodes of $\s{V}$ to the binary property $\{v \in \s{X}\}$ or $\{v \notin \s{X}\}$, i.e., $g_{\s{X}}(p)$ can be represented as a vector $[\sum_{v\in\s{X}}p(v)\, ; \, \sum_{v\notin \s{X}} p(v)]$. 
By a similar reasoning we get
          \[ \left\| f(\o{P}^t\p)-f(\hat{\pilim}) \right\|_{TV}
              \geq \left\|g_{\s{X}} (f(\o{P}^t\p)) - g_{\s{X}}(f(\hat{\pilim}))\right\|_{TV}. \]
If we take $\p = \hat{\pilim}_{\s{C} \times \s{X}}$ as defined in Lemma \ref{lemma:cut} and we use the triangle inequality again, it follows that   
      \begin{align*} 
        &\left\| g_{\s{X}}(f(\o{P}^t\hat{\pilim}_{\s{C} \times \s{X}})) - g_{\s{X}}(f(\hat{\pilim})) \right\|_{TV}\\
          &\quad\geq 
              \left\| g_{\s{X}}(f(\hat{\pilim}_{\s{C}\times \s{X}})) - g_{\s{X}}(f(\hat{\pilim})) \right\|_{TV}
                  - \left\| g_{\s{X}}(f(\hat{\pilim}_{\s{C}\times \s{X}})) 
                      - g_{\s{X}}(f(\hat{\o{P}}^t\hat{\pilim}_{\s{C}\times \s{X}}))\right\|_{TV} \\
          &\quad\geq \frac{1}{2} - \pg{\s{C}\times \s{X}^c}{\o{P}^t\hat{\pilim}_{\s{C}\times \s{X}}}
                      \geq \frac{1}{2} - t\Phi_{\s{C}\times \s{X}}(\o{P}) 
                          = \frac{1}{2} - t\Phi_{\s{X}}(\o{P}_{\s{V}}).
      \end{align*}
From first to second line, we have used that $\pg{x}{\,g_{\s{X}}(f(\hat{\pilim}_{\s{C} \times \s{X}}))\,}=1$, while 
$\pg{x}{\,g_{\s{X}}(f(\hat{\pilim}))\,} \leq 1/2$; this ensures that on the right hand side of the first line, the first term is $\geq 1/2$, while the second term boils down exactly to the probability to be in $\s{X}^c$. The last inequalities follow from Lemma \ref{lemma:cut} and equation \eqref{eq:marg-cond}.
We thus find altogether that 
$$\left\| f(\o{P}^t\p) - f(\hat{\pilim}) \right\|_{TV} \geq \frac{1}{2} - t\Phi_{\s{X}}(\o{P}_{\s{V}}) \, .$$
For $t=\tau(1/4)$, the left hand side must be smaller than $1/4$ and rearranging terms yields $\tau(1/4) \geq 1/(4\Phi_{\s{X}}(\o{P}_{\s{V}}))$. The same obviously holds true when minimizing $\Phi_{\s{X}}$ over $\s{X}$, yielding the statement with $\Phi(\o{P}_{\s{V}})$. The fact that $\Phi_{\pilim} \geq \Phi(\o{P}_{\s{V}})$ was already discussed after equation \eqref{eq:marg-cond}.
  \end{proof}
Combining Theorem \ref{thm:main} with the bound on the mixing time of LMCs provided by Lemma \ref{lem:lift-mixing}, leads to the following bound for quantum channels:
 \begin{thm}[main paper Thm.2, generalized version]
   Any local and invariant stochastic linear map has a mixing time $\tau(1/4) \geq 1/(4\Phi_{\pilim})$. As a consequence, any $\pilim$-invariant QW has a mixing time $\tau(1/4) \geq 1/(4\Phi_{\pilim})$.
   There exists such a QW that has a mixing time
  $\tau(\epsilon) \leq O(\,\log(1/\min_k\pilim_k) \, \log(1/\epsilon){\; /\; \Phi_{\pilim}}\,). $
 \end{thm}
 \begin{proof}
   If the stochastic linear map has a $1/4$-mixing time $\tau(1/4)$, then according to Theorem \ref{thm:main} with $\epsilon_0 \leq 1/4$ we can construct a local LMC with a marginal $1/4$-mixing time equal to $\tau(1/4)$. However, from Lemma \ref{lem:lift-mixing} we can bound the $1/4$-mixing time of any such LMC by the graph conductance $\Phi_{\pilim}$.
At the beginning of the second section we explain how we can associate a stochastic linear map that is local and invariant, to any local quantum channel that leaves the same target distribution invariant. This readily implies that the lower bound on mixing time holds for such QWs. The existence result follows from the same existence result for lifted Markov chains, proven in \cite{chen1999lifting}. Its validity for QWs follows by recognizing that lifted Markov chains are a special class of quantum channels.
 \end{proof}

\subsection{Quantum Walks on lattices}

Consider a $d$-dimensional periodic lattice $\integer_M^d$ of side $M$, encoded in a graph with node set
\[
   \s{V}
      = \{(i_1,i_2,\dots,i_d)|1 \leq i_k \leq M,\forall k\}.
\]
Similar to the QW/LMC construction for the cycle, we lift this graph by adding a set of coin states $\s{C} = \{+_k,-_k|1\leq k\leq d\}$. An LMC on this graph thus takes place on the vector space $\Hi_{\s{C}} \otimes \Hi_{\s{V}} = \sp\{\cl{c,v}|(c,v)\in\s{C}\times\s{V}\}$.
With operator $\o{P}_k^\pm$ defined on $\Hi_{\s{V}}$ as the cyclic permutation of the $k$-th dimension, that is, $\o{P}^\pm_k \cl{\dots,i_{k-1},i_k,i_{k+1}\dots} = \cl{\dots,i_{k-1},(i_k\pm 1)\mathrm{mod }M,i_{k+1},\dots}$ for all $1\leq k\leq d$, the LMC defined in \cite{diaconis2013spectral} writes: 
\begin{align} \label{eq:LMC-torus}
\begin{split}
    \o{P}
      &= \left( \sum_k \cl{+_k}\clt{+_k} \otimes \o{P}^+_k + \cl{-_k}\clt{-_k} \otimes \o{P}^-_k \right)
              \cdot (\o{S} \otimes \id_{\s{V}}) \\
      &= \begin{bmatrix} \o{P}^+_1 \\ & \o{P}^-_1 \\ 
                                      & & \o{P}^+_2 \\ & & & \ddots \\ & & & & \o{P}^-_d \end{bmatrix}
            \cdot (\o{S} \otimes \id_{\s{V}}),
 \end{split}
 \end{align}
where we now specifically select
\[
    \o{S}
      = \begin{bmatrix} 1-(2d-1)\alpha & \alpha & \dots & \alpha \\ \alpha & 1-(2d-1)\alpha & \dots & \alpha \\
                \vdots & \vdots & \ddots & \vdots \\ \alpha & \alpha & \dots & 1-(2d-1)\alpha \end{bmatrix},
\]
with $\alpha = 1/(2dM)$, and we recall that $\id_{\s{V}}$ is the identity matrix on $\Hi_{\s{V}}$.
 
We will prove the mixing time for $M$ odd. For $M$ even, the LMC shows a parity problem: starting from a single state, at any given time the walk will be supported only on the even or only on the odd nodes. This is easily remedied by for instance modifying $\o{P}$ to $(\o{P}+\id_{\s{C}\times\s{V}})/2$, which changes the mixing time only by a constant factor, or by randomizing the parity of the initial state. To facilitate its reading, we again structure the proof using two technical lemmas.

\begin{lem}\label{aslem2}
Assume that the LMC in Eq.~\eqref{eq:LMC-torus} starts with any $\p_0 \in \{\cl{(c,v)}\}$, i.e., with all its weight concentrated on a single vertex and single coin choice $c \in \{k_+,k_-\}$ for some $k$. Then the resulting distribution after $2M$ time steps has uniformly mixed $i_k$ with a probability $\geq 1/(16d)$, in the sense that $\pg{i_k = n}{\p_{2M}} \geq 1/(16d\,M)$ for all $n \in \{1,2,...,M\}$.
\end{lem}

\begin{proof}
By symmetry, we can consider without loss of generality that the initial distribution is
    \[ \p_0 = \cl{+_1,1,\dots,1}. \]
Writing $\o{S} = (1-2d\alpha)\o{I}_{\s{C}} + 2d\alpha \mathbf{1}/{2d}$, where $\mathbf{1}$ denotes the matrix of all ones, we say that at each time step with probability $2d\alpha=1/M$ a coin toss takes place. Then the probability of a single coin toss happening over $2M$ steps is given by
    \begin{align*}
      \pg{1 \text{ coin toss}}{2M \text{ steps}}
        &= \binom{2M}{1} \cdot \frac{1}{M} \cdot \left(1-\frac{1}{M}\right)^{2M-1} \\
        &= 2\, \left(1-\frac{1}{M}\right)^{2M-1} \geq \frac{1}{8}
        \quad \text{for} \;\; M \geq 2,
     \end{align*}
where the inequality follows from the fact that $(1-1/M)^M$ is an increasing function of $M$, going from $1/4$ for $M = 2$ to $1/e$ for $M$ large.
From this, the event $E_1$ that a single coin toss takes place and switches the coin state from $+_1$ to $-_1$, occurs with probability
    \[
      \pg{E_1}{2M \text{ steps}}
          = \frac{1}{2d} \cdot 
                  \pg{1 \text{ coin toss}}{2M \text{ steps}} \geq \frac{1}{16d}.
    \]
When $E_1$ holds true with the single coin toss at time $T \in [1,2M]$, the distribution at time $2M$ equals
  \[
      \p_{2M} 
          = \cl{-_1,\;1+(T-1)-(2M-T+1)\text{mod}M\;,1,\dots,1} 
          = \cl{-_1,\;(2T-1)\text{mod}M\;,1,\dots,1} \; .
  \]
Yet, conditional on the fact that $E_1$ holds true, the timing $T$ of the single coin toss is uniformly distributed between 1 and $2M$. As a consequence, $2T-1$ is uniformly distributed over $1,3,5,...,M,2,4,6,...,M-1,1,3,5,...,M,2,4,6,...,M-1$, i.e., effectively over the integers from $1$ to $M$.
Thus
  \[
      \pg{i_k = n}{\p_{2M}}
          \geq \pg{i_k = n}{E_1} \cdot \pg{E_1}{2M \text{ steps}}
          = \frac{1}{16d} \, \frac{1}{M}\quad \forall n \in \{1,2,...,M\},
  \]
which proves the statement.
\end{proof}

Our application of Lemma \ref{aslem2} to prove the following result is loose. Its sequential use, coordinate by coordinate, leaves further room for improvement, and this is why we think that it should be possible to win a factor $d$ on the mixing time. However, the resulting estimate is sufficient for the story of the main paper.

\begin{lem}\label{aslem1}
Consider the LMC defined in Eq.\eqref{eq:LMC-torus} on $\integer_M^d$, with $M$ odd. For any initial distribution $p_0$, the distribution $p_T$ after $T = 3M\cdot d(d\log(d)+d)$ steps satisfies
    \begin{equation} \label{eq:contraction} 
      \p_T = \o{P}^T\p_0 = q \cdot \pilim + (1-q) \cdot \left(\frac{\p_T-q\cdot\pilim}{1-q} \right),
    \end{equation}
with $\tilde{\p}_T=(\p_T-q\cdot\pilim)\,/\,(1-q)$ a \emph{positive} distribution; $\pilim$ the stationary distribution of $\o{P}$, which is the uniform distribution over $\s{C} \times \s{V}$; and $q =(1-1/e)/2$ where $e=\exp(1)$.
\end{lem}
\begin{proof}
We consider time intervals of $3M$ steps, which we analyze as follows:
    \begin{itemize}
     \item in the first $M$ steps: As in Lemma \ref{aslem2}, we say that at each time step with probability $2d\alpha=1/M$ a coin toss takes place. Now, we use that this completely randomizes the coin state. The probability that no such coin toss has happened after $M$ steps is $(1-1/M)^M \leq 1/e$.
     \item in the next $2M$ steps: By lemma \ref{aslem2} the coordinate corresponding to the randomized coin state is mixed uniformly with a probability $\geq 1/16d$. If that coordinate was already in a more mixed state than in the hypothesis of lemma \ref{aslem2}, then the resulting mixing can only be better.
    \end{itemize}
Each interval of $3M$ steps will thus uniformly mix a random coordinate with probability $r \geq 1/(16e\,d)$.
  Using Cantelli's inequality for a binomial process with success probability $r$, we find the following bound for the number of successful mixing episodes $k$:
    \[ \pg{k \geq l}{\;2l/r \text{ iterations}} \geq 1/2. \]    
So if we go through $(d\log(d)+d) \cdot 32e\, d$ such intervals of $3M$ steps, then with a probability $1/2$ we will have mixed $d\log(d)+d$ coordinates; the latter are chosen randomly according to independent uniform processes with repetition.
According to the coupon collector's problem, $d\log(d)+d$ random choices selects all coordinates with a probability $(1-1/e)$. This implies that we can bound the state after $T = 3M \cdot (d\log(d)+d) \cdot 32 e\, d$ steps as $\; \p_T \geq q \cdot \pi \;$, with $q = (1-1/e)/2$.
\end{proof}\vspace{3mm}

We now have all the pieces to prove the actual result.

\begin{thm}[main paper Thm.3]
  The LMC defined in Eq.\eqref{eq:LMC-torus} on $\integer_M^d$, with $M$ odd, has a mixing time $\tau(\epsilon) \leq\, O(M\; d^2\log(d) \, \log(1/\epsilon))$.  \;
(We recall that the assumption that $M$ is odd is a standard technicality, to avoid discussing all possible easy ways to break the parity symmetry.)
\end{thm}
\begin{proof}
  We have shown in Lemma \ref{aslem1} that for any initial distribution $\p_0$ over $\s{V}$, with a fixed probability $q$ the state will be uniformly mixed after $T = 3M\cdot d(d\log(d)+d)$ steps, i.e., $\p_T$ will be of the form \eqref{eq:contraction}.
Then after $2T$ steps, we get
    \[ 
        \p_{2T} 
            = \o{P}^T\p_T 
            = \o{P}^T(q\cdot \pilim + (1-q)\cdot \tilde{\p}_T) 
            = q\cdot \pilim + q(1-q) \cdot \pilim + (1-q)^2 \cdot \tilde{\p}_{2T}.
    \]
And after another $(k-2)T$ steps we find by an iterative argument that
    \[ \p_{kT} = (1-(1-q)^k) \cdot \pilim + (1-q)^k \cdot \tilde{\p}_{kT}. \]
This shows that
    \[
        \| \o{P}^t\p_0 - \pilim \|_{TV} 
              \leq (1-q)^{\lfloor t/T \rfloor} \quad \forall t \geq 0,\p_0,
    \]
and thus $\;\| \o{P}^t\p_0 - \pilim \|_{TV} \leq \epsilon\;$ provided $\; t \geq T\cdot \left(1 + \frac{\log{\epsilon^{-1}}}{\log{(1-q)^{-1}}} \right) \;$.
As $q$ is a fixed constant below 1, and $T\in O(M\cdot d^2\log(d))$, this proves the claimed mixing time.
\end{proof}

\end{document}